\documentclass{ijuc}
\usepackage[pdftex]{graphics}
\usepackage{amssymb, amscd, amsfonts}
\usepackage{mathrsfs}
\usepackage{amsmath, amsthm}
\usepackage{booktabs}

\newcommand{\god}{G\"odel}

\newcommand{\Z}{{\mathbb{Z}}}

\newcommand{\R}{{\mathbb{R}}}
\newcommand{\V}{{\mathscr{V}}}
\newcommand{\C}{{\mathscr{C}}}
\newcommand{\I}{{\mathscr{I}}}

\newcommand{\F}{{\mathscr{F}}} 
\newcommand{\T}{{\mathscr{T}}} 
\newcommand{\G}{{\mathscr{G}}} 

\newcommand{\Form}{\textsc{Form}}

\newtheorem{theorem}{Theorem}[section]
\newtheorem{lemma}[theorem]{Lemma}

\newtheorem{proposition}[theorem]{Proposition}

\newtheorem{definition}[theorem]{Definition}
\theoremstyle{definition}
\newtheorem*{remark}{Remark}
\newtheorem*{ack}{Acknowledgement}
\newtheorem{example}{Example}

\newcommand{\GL}{\mathbb{G}}

\begin{document}
%
\title{Valuations in G\"{o}del Logic, and the Euler Characteristic}

\author{Pietro Codara
\and Ottavio M. D'Antona
\and Vincenzo Marra\email{marra@dico.unimi.it}
}

\institute{Dipartimento di Informatica e Comunicazione, Universit\`{a} degli Studi di
Milano, Italy
}

\maketitle

\maketitle

\begin{abstract}
Using the lattice-theoretic version of the Euler characteristic introduced by
V.\ Klee and G.-C.\ Rota in the  Sixties, we define the Euler characteristic of a formula in
\god\ logic (over finitely or infinitely many truth-values). We then prove that
the information encoded by the Euler characteristic is classical, i.e.\
coincides with the analogous notion defined over Boolean logic. Building on
this, we define many-valued versions of the Euler characteristic of
a formula $\varphi$, and prove that they indeed
provide information about the logical status of $\varphi$ in \god\
logic. Specifically, our first main result shows that the many-valued Euler
characteristics are invariants that separate many-valued tautologies from
non-tautologies. Further, we offer an initial investigation of the linear
structure of these generalised characteristics. Our second main result is that
the collection of many-valued characteristics forms a linearly
 independent set in the real vector space of all valuations
of
 \god\ logic over finitely many propositional
variables.
\end{abstract}

\keywords{\god{} logic, \god{} algebra, distributive lattice,
Euler characteristic, valuation, vector space of valuations.}

\section{Introduction and Background}\label{s:intro}
Some decades ago, V. Klee and G.-C.\ Rota introduced a lattice-theoretic
analogue of the Euler characteristic, the celebrated
topological invariant of polyhedra. Let us recall their definition. Let $L$ be a
distributive
lattice. A function $\nu \colon L \to \R$ is a
\emph{valuation} if it satisfies
\begin{equation}\label{eq:valuation}
 \nu(x)+\nu(y)=\nu(x\vee y)+\nu(x\wedge y)
\end{equation}
for all $x,y,z \in L$.  Recall that an element $x \in L$ is
\emph{join-irreducible} if it is not the bottom element of $L$, and $x = y
\vee z$ implies $x=y$ or $x=z$ for all $y,z \in L$. When $L$ is finite, it turns
out \cite[Corollary 2]{rota} that any valuation $\nu$
is uniquely determined by its values on the join-irreducible elements of
$L$, along with its value at the bottom element $\bot$ of $L$.

 \medskip
\begin{definition}[{\cite[p.\ 120]{klee}, \cite[p.\ 36]{rota}}] The Euler
characteristic of a finite
distributive lattice $L$ is the unique  valuation $\chi \colon L \to
\R$ such
that $\chi(x)=1$ for any join-irreducible element $x \in L$, and $\chi(\bot)=0$.
\end{definition}
\medskip

\emph{{\god} \textup{(}infinite-valued propositional\textup{)} logic}
$\GL_\infty$ \cite{dummett}
 can be syntactically defined as the schematic extension of the
intuitionistic  propositional calculus by the \emph{prelinearity axiom} $(\alpha
\to \beta) \vee (\beta \to \alpha)$.
It can also be semantically defined as a many-valued  logic \cite{hajek}, as follows.
Write $\Form$ for the set of   formul\ae\ over propositional variables
$X_1,X_2,\dots$ in the language $\wedge,\vee,\to,\neg,\bot,\top$. (Here, $\bot$
and $\top$
are the logical constants {\it falsum} and {\it verum}, respectively.)
An
\emph{assignment} is a function $\mu \colon \Form \to [0,1] \subseteq \R$
with values in the real unit interval
such that, for any two  $\alpha,\ \beta \in \Form$,

\begin{itemize}
        \item[] $\mu(\alpha \wedge \beta) = \min\{\mu(\alpha),\mu(\beta)\}$
        \item[] $\mu(\alpha \vee \beta) = \max\{\mu(\alpha),\mu(\beta)\}$
        \item[] $\mu(\alpha \rightarrow \beta) = \left\{\begin{array}{l}
        1\ \ \ \ \ \ \,\textrm{if}\ \mu(\alpha)\leq\mu(\beta)\\
        \mu(\beta)\ \ {\rm otherwise}
        \end{array}\right.$
\end{itemize}
and $\mu(\neg \alpha)=\mu(\alpha \to \bot)$, $\mu(\bot)=0$, $\mu(\top)=1$.
A \emph{tautology} is a formula $\alpha$ such that $\mu(\alpha)=1$ for every
assignment $\mu$. As is well known, {\god} logic
is complete with respect to this many-valued semantics. Indeed, for $\alpha \in
\Form$, let
us write $\vdash \alpha$
to mean that $\alpha$ is derivable from the axioms of $\GL_\infty$ using
\textit{modus ponens} as the only
deduction rule. Then the
completeness theorem guarantees that $\vdash \alpha$ holds if and only if
$\alpha$ is a tautology. A stronger result holds: like classical
logic, $\GL_\infty$ also enjoys completeness for theories.
For proofs and background,  see  \cite{hajek}.

For an integer $n\geq 1$, let us write $\Form_n$ for the set of all
formul\ae\ whose propositional variables are contained in
$\{X_1,\ldots,X_n\}$. As
usual, $\varphi, \psi \in \Form_n$ are called
\emph{logically equivalent} if both $\vdash \varphi \to \psi$ and $\vdash \psi
\to \varphi$
hold. Logical equivalence is an equivalence
relation, written $\equiv$, and its equivalence classes are denoted
$[\varphi]_\equiv$. By a routine check, the quotient set $\Form_n/\equiv$ endowed with
operations
$\wedge$, $\vee$, $\top$, $\bot$ induced from the corresponding logical
connectives becomes a distributive lattice with top and bottom element $\top$
and $\bot$, respectively. When $\Form_n/\equiv$ is further
endowed with the operation $\to$ induced by implication, it becomes a Heyting
algebra satisfying prelinearity; such algebras we call \emph{\god\ algebras}
(cf.\ the term \emph{G-algebras} in \cite[4.2.12]{hajek}). The specific \god\
algebra $\G_n = \Form_n/\equiv$ is, by construction, the \emph{Lindenbaum
algebra} of \god\ logic over the language $\{X_1,\ldots,X_n\}$.

\smallskip It is a remarkable fact due to Horn \cite[Theorem 4]{horn} that $\G_n$ is finite
for each integer $n \geq 1$, in analogy with Boolean algebras. A second
important fact is that a finite Heyting algebra is a \god\ algebra if and only
if its collection of join-irreducible elements, ordered by restriction from
$\G_n$,  is a \emph{forest}; i.e.\ the lower
bounds of any such element are a totally ordered set. A more general version of
this result
is also due to Horn \cite[Theorem 2.4]{horn2}.

\smallskip Knowing that $\G_n$ is a finite distributive lattice whose elements
are formul\ae\ in $n$ variables, up to logical equivalence, one is led to give
the following definition.

\medskip
\begin{definition}\label{d:Euler}The \emph{Euler
characteristic} of a formula $\varphi \in \Form_n$, written
\[
 \chi(\varphi) \ ,
\]
 is
the number $\chi([\varphi]_\equiv)$, where $\chi$ is the Euler
characteristic of
the finite distributive lattice $\G_n$.
\end{definition}
\medskip

\noindent However, the question is now whether $\chi(\varphi)$
encodes genuinely \emph{logical} information about $\varphi$, just like the
Euler characteristic of a polyhedron provides geometric information about that
polyhedron. The answer turns out to be affirmative. As
usual, we say that an
assignment $\mu \colon \Form_n
\to [0,1]$ is \emph{Boolean} if it takes values in $\{0,1\}$.

\medskip
\begin{theorem}\label{t:boole}Fix an integer $n \geq 1$.
For any formula $\varphi
\in \Form_n$, the Euler
characteristic $\chi(\varphi)$ equals the number of Boolean assignments $\mu
\colon \Form_n \to \{0,1\}$ such that $\mu(\varphi)=1$.
\end{theorem}
\medskip

\noindent Theorem \ref{t:boole} will turn out to be an easy corollary of our first main result,
Theorem \ref{t:main}.
As an immediate consequence of Theorem \ref{t:boole},
\[
 0 \leq \chi(\varphi) \leq 2^n
\]
for any $\varphi \in \Form_n$.
In particular, note that the following hold.
\begin{itemize}
 \item If $\varphi$ is a tautology in $\GL_\infty$, then
$\chi(\varphi)=2^n$.
\item If $\chi(\varphi)=2^n$, then  $\varphi$ is
 a tautology in classical propositional logic.
\item If $\chi(\varphi)=0$, then $\varphi$ is a contradiction in classical
propositional logic, and conversely.
\end{itemize}
In summary, Theorem \ref{t:boole} shows that, while $\chi(\varphi)$ does encode non-trivial
logical information, that information is classical, and independent of \god\
logic. In fact, if one replicates the above construction over classical logic,
one ends up with a valuation $\chi$ on the Boolean algebra of
$n$-variable formul\ae\ that simply counts the number of atoms below each
element in the Boolean algebra.
By the same token, the Euler characteristic cannot tell apart the tautologies in
G\"{o}del
logic from the remaining formul\ae, whereas it does so for classical
tautologies. In Section \ref{s:main} we show how to remedy this by
considering different valuations on $\G_n$ which we refer to as \emph{generalised
characteristics} (Definition \ref{d:chik}). As it will emerge, they can be thought of as many-valued variants of the
classical  characteristic of Definition \ref{d:Euler}.

Our first main result, Theorem \ref{t:main}, shows
that $\chi_k$ is a natural generalisation of $\chi$ in that it tells apart the
tautologies in \emph{\god\ $(k+1)$-valued logic $\GL_{k+1}$} from the remaining
formul\ae. Here we recall that $\GL_{k+1}$ is the schematic extension of
$\GL_\infty$ via
\begin{equation}\label{eq:godelkval}
\alpha_1\vee(\alpha_1\to\alpha_2)\vee\cdots\vee(\alpha_1\wedge\cdots\wedge\alpha
_{k} \to \alpha_{k+1})\,.
\end{equation}
\noindent Alternatively, using \cite[Proposition 4.18]{axiom2}, one can equivalently
replace\footnote{We thanks the anonymous referee for bringing \cite{axiom2} to our attention.}
(\ref{eq:godelkval}) by the axiom
\begin{equation*}
\bigvee_{1 \leq i \leq k}\left(\alpha_i\to\alpha_{i+1}\right)\,.
\end{equation*}

Semantically,  restrict assignments to
those taking values
in the set
\[
V_{k+1}=\{0=\frac{0}{k}, \frac{1}{k}, \ldots, \frac{k-1}{k}, \frac{k}{k}=1\}
\subseteq [0,1] \ ,
\]
that is, to \emph{$(k+1)$-valued assignments}.
A tautology of  $\GL_{k+1}$
is defined as a formula that takes value $1$ under
any such assignment. Then $\GL_{k+1}$ is complete with
respect to this semantics; see e.g.\ \cite{bcf} for further
background.

\medskip In Section \ref{s:linear}, we analyse the linear strucutre of the generalised characteristics
introduced in Section \ref{s:main}. The set of valuations over a finite distributive lattice $L$ carries a natural structure of
(real) vector space. This is because the function $\zeta\colon L\to \R$ such that $\zeta(x)=0$ for
each $x\in L$ is a valuation,
and if  $\nu_1,\nu_2 \colon L \to \R$ are valuations, then so is the function $r_1\nu_1+r_2\nu_2$
defined by
\[
(r_1\nu_1+r_2\nu_2)(x)=r_1\nu_1(x)+r_2\nu_2(x) \ \ \text{for each } x \in L\ ,
\]
for any
two real numbers $r_1,r_2\in\R$. It is therefore natural to ask what linear relations
are satisfied by the generalised characteristics. As we prove in our second main result
(Theorem \ref{t:basis})
the answer is none.
\section{The Many-valued Characteristic of a Formula}\label{s:main}
The
\emph{height} of a join-irreducible $g \in \G_n$ is the length
$l$ of the longest chain $g=g_1 > g_2 > \cdots
> g_{l}$  in $\G_n$ with each $g_i$ a
join-irreducible element. We write  $h(g)$ for the height of $g$.

We can now define the generalised characteristics that feature in Section \ref{s:intro}.
\medskip
\begin{definition}\label{d:chik}Fix integers $n, k \geq 1$. We write $\chi_k
\colon \G_n \to \R$ for the unique
valuation on $\G_n$  that  satisfies
\[
 \chi_k(g)=\min{\{h(g),k\}}
\]
for each join-irreducible element $g \in \G_n$, and such that, moreover,
$\chi_k(\bot)=0$.
Further, if $\varphi \in \Form_n$, we define $\chi_k(\varphi)=\chi_k([\varphi]_\equiv)$.
\end{definition}

\medskip
\noindent Clearly, $\chi_1$ is the Euler characteristic $\chi$ of $\G_n$.
  We now need to recall a  notion (cfr.\ \cite[Definition 2.1]{ijar}) that is central
to \god\ logic.

\medskip
\begin{definition}
\label{def:n-equivalent}
Fix integers $n, k \geq 1$. We say that two $(k+1)$-valued assignments  $\mu$
and $\nu$ are \emph{equivalent over the first
$n$ variables}, or \emph{$n$-equivalent}, written $\mu \equiv^k_n \nu$,
if and only if there exists a permutation
$\sigma \colon \{1,\ldots,n\} \to \{1,\ldots,n\}$
such that
\begin{equation}
\begin{split}
0 \preceq_0 \mu(X_{\sigma(1)}) \preceq_1 \cdots \preceq_{n-1} \mu(X_{\sigma(n)})
\preceq_{n} 1\ ,\\
0 \preceq_0 \nu(X_{\sigma(1)}) \preceq_1 \cdots \preceq_{n-1}
\nu(X_{\sigma(n)}) \preceq_{n} 1\ ,
\end{split}
\label{eq:assignment}
\end{equation}
where $\preceq_i\ \in \{<,=\}$, for $i=0,\dots,n$.
\end{definition}
\medskip
Thus, two equivalent assignments induce the same strict inequalities ($<$) and
equalities ($=$) on
the propositional variables. Clearly, $\equiv^k_n$ is an equivalence relation.
In various guises, the above notion of equivalent assignments plays a crucial
r\^{o}le in the investigation of G\"{o}del logic; see e.g.\ \cite{ijar, jlc}.
For our purposes here, we observe that  distinct
$2$-valued (=Boolean)
assignments are
not equivalent, so that there are $2^n$ equivalence classes of such
assignments over the first $n$ variables.

\medskip
We next introduce  the
$(k+1)$-valued analogue of $2^n$. As will be proved in Subsection
\ref{s:proof}, the following recursive formula counts the number of join-irreducible elements
of $\G_n$ having height smaller or equal than $k$.
\begin{equation}
\label{eq:taut}
P(n,k)=\sum_{i=1}^k\sum_{j=0}^n\binom{n}{j}T(j,i)\,, \tag{*}
\end{equation}
where
\begin{equation*}
T(n,k)=
\begin{cases}
        1&\text{if \(k=1\)} ,\\[1.5ex]
        0&\text{if \(k>n+1\)} ,\\[1.5ex]
        \displaystyle{\sum_{i=1}^n{\binom{n}{i}T(n-i,k-1)}}&\text{otherwise} .
\end{cases}
\end{equation*}
\begin{table}[ht!]\small
\begin{center}
\resizebox{\columnwidth}{!}{%
\textrm{\begin{tabular}{r | *{7}c}
     & k=1 & 2 & 3 & 4 & 5 & 6 & 7 \\
\midrule
 n=1 & 2 & 3 & 3 & 3 & 3 & 3 & 3 \\
 2   & 4 & 9 & 11 & 11 & 11 & 11 & 11 \\
 3   &8 & 27 & 45 & 51 & 51 & 51 & 51 \\
 4   &16 & 81 & 191 & 275 & 299 & 299 & 299 \\
 5   &32 & 243 & 813 & 1563 & 2043 & 2163 & 2163 \\
 6   &64 & 729 & 3431 & 8891 & 14771 & 18011 & 18731 \\
 7   &128 & 2187 & 14325 & 49731 & 106851 & 158931 & 184131 \\
 8   &256 & 6561 & 59231 & 272675 & 757019 & 1407179 & 1921259 \\
 9   &512 & 19683 & 242973 & 1468203 & 5228043 & 12200883 & 20214483 \\
\end{tabular}}
}
\end{center}
\caption{The number of distinct equivalence classes of
$(k+1)$-valued assignments over $n$ variables.}
\label{T: F2}
\end{table}

Our aim in this section is to establish the following result.
\medskip
\begin{theorem}\label{t:main}Fix integers $n, k \geq 1$, and a formula
$\varphi
\in \Form_n$.
\begin{enumerate}
\item $\chi_k(\varphi)$ equals the  number of
$(k+1)$-valued assignments $\mu
\colon \Form_n \to [0,1]$ such that $\mu(\varphi)=1$, up to $n$-equivalence.
\item $\varphi$ is a tautology in $\GL_{k+1}$ if and only if
$\chi_{k}(\varphi)=P(n,k)$.
\item $\varphi$ is a tautology in $\GL_\infty$ if and only if it is a
tautology in $\GL_{n+2}$ if and only if
$\chi_{n+1}(\varphi)=P(n,n+1)$.
\end{enumerate}
\end{theorem}
\medskip

Since distinct Boolean
assignments are
pairwise inequivalent, Theorem \ref{t:boole} is an immediate consequence of Theorem
\ref{t:main}. We note that
\[
P(n,1)=\sum_{j=0}^n\binom{n}{j}T(j,1)=\sum_{j=0}^n\binom{n}{j}=2^n\,,
\]
so that $P(n,k)$ indeed is the $(k+1)$-valued analogue of $2^n$.
\begin{remark}A closed formula for the number $P(n,k)$ may be obtained combining the
results of
\cite{nelsenetal} on the number of chains in a power set, and the results of \cite[Subsection 4.2]{dm} relating the number of
join-irreducible elements of $\G_n$ to ordered partitions of finite sets.
We  do not provide the combinatorial details in the present paper.
\end{remark}
\subsection{Proof of Theorem \ref{t:main}}\label{s:proof}
\subsubsection{Proof of \textup{(\ref{eq:taut})}}\label{sub:taut}

Let $\F_n$ be the forest of join-irreducible elements of $\G_n$, and let $\T_n$ be the unique
tree of $\F_n$ having maximum height (cfr.\ \cite[Section 2.3]{jlc}).
Here, by the height $h(F)$ of a forest $F$
we mean the cardinality of its longest chain. Denote by $\uparrow g$ the upper set
of an element $g$, that is,
\begin{equation*}
\uparrow g = \left\{x \in F\, |\, x \geq g \right\} .
\end{equation*}
Similarly, the lower set of $g$ is
\begin{equation*}
\downarrow g = \left\{x \in F\, |\, x \leq g \right\} .
\end{equation*}
The height of an element $g \in F$ is the height of $\downarrow g$.
Recall that an \emph{atom} of a partially ordered set with minimum is an element that
covers its minimum.
It can be shown (cfr.\ \cite[Lemma 2.3 -- (a)]{jlc}) that $\T_n$ has precisely $\binom{n}{i}$ atoms $a$ with
$\uparrow a \cong \T_{n-i}$, for each $i=1,\ldots,n$, and no other atom.
Observing that $\T_0$ is the one-element tree, and that $h(\T_n)=h(\T_{n-1})+1$ for each $n$,
we immediately obtain the following recursive formula
for the number of elements of $\T_n$ having height $k$.
\begin{equation*}
T(n,k)=
\begin{cases}
        1&\text{if \(k=1\)}\,,\\[1.5ex]
        0&\text{if \(k>n+1\)}\,,\\[1.5ex]
        \displaystyle{\sum_{i=1}^n{\binom{n}{i}T(n-i,k-1)}}&\text{otherwise}\,.
\end{cases}
\end{equation*}
Further, $\F_n$ contains precisely $\binom{n}{i}$ distinct copies of $\T_i$, for $i=0,\ldots,n$, and no other tree (cfr.\ \cite[Lemma 2.3 -- (b)]{jlc}).
Thus, as claimed, $P(n,k)$ gives the number of elements of $\F_n$ having height smaller or equal than $k$ (i.e.\ the number of join-irreducible elements
of $\G_n$ having height smaller or equal than $k$).

\subsubsection{Two lemmas}\label{sub:lemma}
\begin{lemma}\label{l:lemma}
Fix integers $n,k \geq 1$, let $x \in \G_n$ and consider the valuation
$\chi_k \colon \G_n \to \R$. Then, $\chi_k(x)$ equals the number of
join-irreducible elements $g \in \G_n$ such that $g \leq x$ and $h(g)\leq k$.
\end{lemma}
\begin{proof}
If $x=\bot$ then, by Definition \ref{d:chik}, $\chi_k(x)=0$, and the Lemma trivially holds.

Let $F$ be the forest of all join-irreducible elements $g \in \G_n$ such that $g \leq x$.
(Recall that $x$ is the join of the join-irreducible elements $g \in F$.)
We proceed by induction on the structure of $F$. If $F$ is the one-element forest, then $x$ is a join-irreducible element,
and $F=\{x\}$. By Definition \ref{d:chik}, $\chi_k(x)=1$, for each $k \geq 1$, as desired.

\medskip
Let now $|F|>1$. Let $l\in F$ be a maximal element of $F$, and consider the forest $F^{-}=F \setminus \{l\}$.
Let $x^{-}$ be the join of the elements of $F^-$. We immediately observe that $x = l \vee x^-$.

\medskip
If $l$ is an atom of $\G_n$, then $l \wedge x^- = \bot$. By (\ref{eq:valuation}) and Definition \ref{d:chik},
$\chi_k(x)=\chi_k(l \vee x^-)=\chi_k(l)+\chi_k(x^-)-\chi_k(l \wedge x^-)=1+\chi_k(x^-)$. Using the inductive hypotheses
on $F^-$ we obtain our statement, for the case $h(l)=1$.

\medskip
Let, finally, $h(l)>1$. Consider the element $l^- = l \wedge x^-$. Let $L$ be
the forest of all join-irreducible elements $g \in \G_n$ such that $g \leq l$, and let $L^-$ be
the forest of all join-irreducible elements $g \in \G_n$ such that $g \leq l^-$.
Since $l$ is a join-irreducible, $L$ is a chain. Moreover, one easily sees that
$L^-=L \setminus \{l\}$.
For a forest $P$, we denote by $|P|_k$ the number of elements $p$ of $P$ such that $h(p)\leq k$.
We consider two cases.

\smallskip
$h(l)\leq k$. We observe that $|F^-|_k = |F|_k - 1$ and that $|L^-|_k = |L|_k - 1$.
Using (\ref{eq:valuation}) and the inductive hypotheses we obtain
$\chi_k(x)=\chi_k(l)+\chi_k(x^-)-\chi_k(l \wedge x^-)=|L|_k + |F|_k - 1 - (|L|_k - 1)=|F|_k$.
In other words, $\chi_k(x)$ equals the number of
join-irreducible elements $g \in \G_n$ such that $g \leq x$ and $h(g)\leq k$.

\smallskip
$h(l)> k$. In this case, we observe that $|F^-|_k = |F|_k$ and that $|L^-|_k = |L|_k$.
Using (\ref{eq:valuation}) and the inductive hypotheses we obtain
$\chi_k(x)=\chi_k(l)+\chi_k(x^-)-\chi_k(l \wedge x^-)=|L|_k + |F|_k - |L|_k=|F|_k$.
In other words, $\chi_k(x)$ equals the number of
join-irreducible elements $g \in \G_n$ such that $g \leq x$ and $h(g)\leq k$, and the lemma is proved.
\end{proof}

\medskip
\begin{lemma}\label{l:2}Fix integers $n,k\geq 1$, and let $\varphi \in \Form_n$.
Let $O(\varphi,n,k)$ be the set of
 equivalence classes $[\mu]_{\equiv_n^k}$ of ($k+1$)-valued assignments
such that $\mu(\varphi)=1$. Further, let $J(\varphi,n,k)$ be the set of
join-irreducible elements $x \in \G_n$ such that $x \leq [\varphi]_\equiv$ and $h(x)\leq k$.
Then there is a bijection between $O(\varphi,n,k)$ and $J(\varphi,n,k)$.
\end{lemma}
\begin{proof}In the proof of this lemma we use techniques  from
algebraic logic; for all unexplained notions, please see \cite{hajek}.

 Fix a ($k+1$)-valued assignment $\mu\colon \Form_n \to V_{k+1}$. Endow
$V_{k+1}$ with its unique structure of \god\ algebra compatible with the total
order of the elements of $V_{k+1}\subseteq [0,1]$.  Then there is a
unique homomorphism of \god\ algebras $h_\mu\colon \G_n \to V_{k+1}$
corresponding to $\mu$, namely,
\begin{equation}\label{e:htomu}
 h_\mu([\varphi]_\equiv)= \mu(\varphi) \ .
\end{equation}
Conversely, given any such homomorphism $h \colon \G_n \to V_{k+1}$, there is
a unique $(k+1)$-valued assignment $\mu_h \colon \Form_n \to V_{k+1}$
corresponding to $h$, namely,
\begin{equation}\label{e:mutoh}
\mu_h(\varphi)=h([\varphi]_\equiv) \ .
\end{equation}
Clearly, the correspondences (\ref{e:htomu}--\ref{e:mutoh}) are mutually
inverse, and thus yield a bijection between $(k+1)$-valued assignments to
$\Form$ and homomorphisms $\G_n\to V_{k+1}$. Further, upon noting that
$\mu_h(\varphi)=1$ in (\ref{e:mutoh}) if and only if $h_\mu([\varphi])=1$ in
(\ref{e:htomu}),
we see that this bijection restricts to a bijection
\begin{equation}\label{e:bij1}
O'(\varphi,n,k) \cong \hom{(\varphi, \G_n,V_{k+1})}
\end{equation}
where the right-hand side is the set of homomorphisms $h\colon \G_n\to
V_{k+1}$ such that $h([\varphi]_\equiv)=1$, while the left-hand side is the
collection of $(k+1)$-valued assignments $\mu\colon \Form_n\to V_{k+1}$ with
$\mu(\varphi)=1$.
Now recall that to any homomorphism $h\colon \G_n\to
V_{k+1}$ one associates the prime (lattice) filter of $\G_n$ given by
$\mathfrak{p}_h=h^{-1}(1)$. Conversely, given a prime filter $\mathfrak{p}$ of
$\G_n$
there is a natural onto quotient map $h_{\mathfrak{p}}\colon \G_n
\twoheadrightarrow \G_n/\mathfrak{p}$, where
$C=\G_n/\mathfrak{p}$ is a chain of finite cardinality; further, $|C|\leq k+1$
if and only if $\mathfrak{p}$ has height $\leq k$, meaning that the chain of
prime filters containing it has cardinality $k$. Since any chain with
$|C|\leq k+1$ embeds into $V_{k+1}$, this shows that each prime filter
$\mathfrak{p}$ of $\G_n$ having height $\leq k$ induces by
\begin{equation}\label{e:split}
h_{\mathfrak{p}}^e\colon \G_n \twoheadrightarrow \G_n/\mathfrak{p}
\overset{e}{\hookrightarrow} V_{k+1}
\end{equation}
one homomorphism (not necessarily onto) $h_{\mathfrak{p}}^e$ from $\G_n$ to
$V_{k+1}$ for \emph{each} choice of the embedding $e \colon \G_n/\mathfrak{p}
{\hookrightarrow} V_{k+1}$. It is now easy to check that two $(k+1)$-valued
assignments  $\mu, \nu\colon \Form_n \to V_{k+1}$
satisfy $\mu\equiv_n^k\nu$ if and only if the associated homomorphisms $h_\mu,
h_\nu$ as in (\ref{e:htomu}) factor as in (\ref{e:split}) for the \emph{same}
prime filter $\mathfrak{p}$, although for possibly different embeddings $e$ and
$e'$ into $V_{k+1}$. It is clear that this yields an equivalence relation on
such homomorphisms $h_\mu,h_\nu$. Let us denote by $\hom_\equiv{(\varphi,
\G_n,V_{k+1})}$ the set of equivalence classes of those homomorphisms $h_\mu$
satisfying $h_\mu([\varphi]_\equiv)=1$. Summing up, from the bijection in
(\ref{e:bij1}) we obtain a bijection
\begin{equation}\label{e:bij2}
 O(\varphi,n,k) \cong \hom_\equiv{(\varphi,
\G_n,V_{k+1})} \ .
\end{equation}
To complete the proof, observe that
since $\G_n$ is finite,
every  filter $\mathfrak{p}$ of $\G_n$ is principal, i.e.\ if
there is an element $p \in \G_n$ such that $\mathfrak{p}=\uparrow p$;
moreover, $\mathfrak{p}$ is prime if and only if  $p$ is join-irreducible.
In other words, there is a bijection between join-irreducible elements and
prime filters of
$\G_n$. By definition, the natural quotient map $\G_n
\twoheadrightarrow \G_n/\mathfrak{p}$ sends $[\varphi]_\equiv$ to $1$ if
and only if $[\varphi]_\equiv$ lies in the prime filter $\mathfrak{p}$; that
is, if and only if $[\varphi]_\equiv \geq p$ in $\G_n$. Moreover, the following
is easily checked. Suppose $\mathfrak{p}=\uparrow p$ as in the above, and let
$\G_n/\uparrow p$ be the quotient algebra, which is a chain because
$\mathfrak{p}$ is prime. Then $|\G_n/\uparrow p|\leq k+1$ if and only if the
height of $p$ satisfies $h(p)\leq k$.  Using the preceding observations,
from (\ref{e:split}) and the definition of $\hom_\equiv{(\varphi,
\G_n,V_{k+1})}$ we  obtain a bijection
\begin{equation}\label{e:bij3}
\hom_\equiv{(\varphi,
\G_n,V_{k+1})} \cong J(\varphi,n,k) \ .
\end{equation}
The lemma follows from (\ref{e:bij2}) and (\ref{e:bij3}).
\end{proof}

\subsection{End of Proof of Theorem \ref{t:main}}
\begin{enumerate}
\item By Lemma \ref{l:lemma} the value $\chi_k(\varphi)=\chi_k([\varphi]_\equiv)$ is given by the
number of join-irreducible elements $g \in \G_n$ such that $g \leq [\varphi]_\equiv$
and $h(g)\leq k$. By Lemma \ref{l:2}, such number equals the number of equivalence
classes $[\mu]_{\equiv_n^k}$ of ($k+1$)-valued assignments
such that $\mu(\varphi)=1$, and the statement follows.\smallskip
\item As proved in Subsection \ref{sub:taut}, the formula $P(n,k)$ counts the total number of join-irreducible elements of
$\G_n$ having height smaller or equal than $k$. By Lemma \ref{l:lemma}, $\chi_{k}(\varphi)=P(n,k)$ if and only if
all the join-irreducible elements $g \in \G_n$ such that $h(g)\leq k$ satisfy $g \leq [\varphi]_\equiv$. By Lemma \ref{l:2},
the latter holds if and only if each ($k+1$)-valued assignment $\mu \colon \Form_n \to [0,1]$
satisfies $\mu(\varphi)=1$, i.e.\ $\varphi$ is a tautology in $\GL_{k+1}$, as desired.\smallskip
\item \textit{Claim:}
     If $\varphi \in \Form_n$ is a tautology in $\GL_{n+2}$, then it is a tautology in $\GL_\infty$.

     \smallskip\noindent
     \textit{Proof of Claim:} Suppose, by way of contradiction, that $\varphi$ is not a
     tautology in $\GL_\infty$, but it is a tautology in $\GL_{n+2}$.
     Thus, there must exists an assignment $\mu$ such that $\mu(\varphi)<1$. An easy
     structural induction shows that $\mu(\varphi) \in \{0, \mu(X_1), \ldots, \mu(X_n), 1\}$.
     But then, the restriction of $\mu$ onto its image yields an ($n+2$)-valued
     assignment $\bar\mu$ such that $\bar\mu (\varphi) < 1$, a contradiction.

\smallskip
As one can immediately check, if $\varphi$ is a tautology in $\GL_\infty$, then it is
a tautology in $\GL_{n+2}$. Thus, using the Claim, $\varphi$ is a tautology
in $\GL_{n+2}$ if and only if it is a tautology in $\GL_\infty$.
Finally, by  statement 2) of this theorem , $\varphi$ is a tautology
in $\GL_{n+2}$ if and only if $\chi_{n+1}(\varphi)=P(n,n+1)$, and the last statement of the theorem is proved.
\end{enumerate}

\begin{example}\label{s:end}
Let us consider the \god{} algebra $\G_1$, depicted in Figure \ref{F: chi}.
Lemma \ref{l:lemma} allows us to compute the values
of $\chi_k(x)$ for each $x\in \G_1$, simply by counting the number
of join-irreducible elements under $x$ having height not greater than $k$.
The results are displayed in Figure \ref{F: chi}, for $k = 1$ (i.e.\ for the Euler characteristic),
and for $k=2$. Note that for $k \geq 3$ and for each $x \in \G_1$, $\chi_k(x)$ and $\chi_2(x)$ coincide,
by  statement 3 in Theorem \ref{t:main}.
\begin{figure}[ht!]
        \begin{center}
                \includegraphics{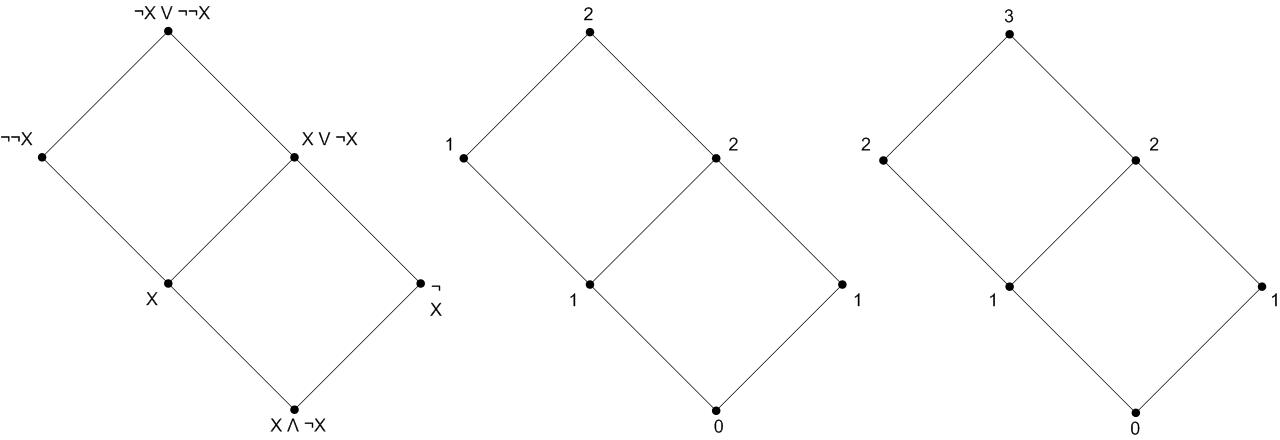}
        \end{center}
        \caption{The \god{} algebra $\G_1$ (left), and the values of $\chi_1$ (middle) and $\chi_2$ (right).}
        \label{F: chi}
\end{figure}

Let us consider the formula $\neg\neg X$. One can check that, up to $n$-equivalence, there
are two distinct $3$-valued assignments $\mu,\nu \colon \Form_1 \to \{0, \frac{1}{2}, 1\}$ such that
$\mu(\neg\neg X)=\nu(\neg\neg X)=1$. Namely, we can take $\mu$ such that $\mu(X)=1$, and $\nu$ such that $\nu(X)=\frac{1}{2}$.
In fact, as one sees in Figure \ref{F: chi}, $\chi_2(\neg\neg X)=2$. The assignment $\mu(X)$ is the only Boolean assignment such that
$\mu(\neg\neg X)=1$. Actually, $\chi_1(\neg\neg X)=\chi(\neg\neg X)=1$.
\end{example}

\section{The linear structure of the characteristics}\label{s:linear}
Let $J_1,\ldots,J_{u_n}$ display all
join-irreducible elements of $\G_n$, for an integer $n\geq 1$. For $i=1,\ldots,u_n$,
let  $e_i$ be the unique valuation
of $\G_n$ such that $e_i(J_i)=1$, and $e_i(J_j)=0$ if $j\neq i$. A moment's reflection shows
that $\{e_1,\ldots,e_{u_n}\}$ is a basis of the vector space of all valuations of $\G_n$.
Hence, $u_n$ is the dimension of this space. Let us remark that it follows from the proof of
(\ref{eq:taut}) in Subsection \ref{s:proof} that
\[
 u_n=P(n,n+1) \ .
\]

An \emph{automorphism} of $\G_n$ is a bijective homomorphism of distributive lattices
$\alpha \colon \G_n \to \G_n$. Such homomorphism is then automatically a homomorphism of Heyting ({\it a fortiori}
G\"{o}del)
algebras, too. A valuation $\nu \colon \G_n\to\G_n$ is \emph{invariant} (\emph{under the
automorphisms of $\G_n$}) if
\[
\nu(x)=\nu(\alpha(x)) \ \ \text{for all } x \in \G_n \ ,
\]
where $\alpha$ is an arbitrary automorphism of $\G_n$.
The invariant valuations of any finite distributive lattice form a vector subspace of the vector
space of all valuations, as one checks easily.

\begin{definition}We denote by $\V_n$ the vector space of all valuations of $\G_n$, for an integer
$n\geq 1$. We further denote by $\I_n$ the vector subspace of $\V_n$ consisting of all invariant valuations
of $\G_n$. Finally, we write $\C_n$ for the vector subspace of $\V_n$ generated by the generalised
characteristics $\{\chi_1,\chi_2,\ldots,\chi_{n+1}\}$.
\end{definition}
By definition, then, $\I_n,\C_n\subseteq \V_n$. More is true.
\begin{proposition}For each integer $n\geq 2$, $\C_n\subset \I_n\subset \V_n$. Further, $\C_1\subset \I_1=\V_1$.
\end{proposition}
\begin{proof}
Let us exhibit a non-invariant valuation of $\G_n$ for each $n\geq 2$. Consider the formul\ae\
\begin{align*}
\varphi &= X_1\wedge \neg X_2 \wedge \neg X_3 \wedge \cdots \wedge \neg X_n\ , \\
\psi & = \neg X_1\wedge X_2 \wedge \neg X_3 \wedge \cdots \wedge \neg X_n \ .
\end{align*}
It
can be checked that $[\varphi]_\equiv$ and $[\psi]_\equiv$ are join-irreducible elements \textup{(}in fact, atoms\textup{)}
of $\G_n$. Consider the valuation $\nu\colon\G_n\to\G_n$ such that $\nu([\varphi]_\equiv)=1$, while $\nu(x)=0$
for every other join-irreducible $x \in\ \G_n$. The permutation
\begin{align*}
X_1 &\mapsto X_2 \ , \\
X_2 &\mapsto X_1 \ , \\
X_i &\mapsto X_i \  \text{ for each } i=3,\ldots,n
\end{align*}
uniquely extends to
an automorphism $\alpha$ of $\G_n$. By construction, $\alpha([\varphi]_\equiv)=[\psi]_\equiv$. But then
$\nu([\varphi]_\equiv)=1\neq 0=\nu([\psi]_\equiv)=\nu(\alpha([\varphi]_\equiv))$.
This shows that $\I_n \subset \V_n$ when $n\geq 2$. On the other hand, direct inspection of $\G_1$ (cf.\ Example
\ref{s:end}) shows that the automorphism group of $\G_1$ is trivial, i.e.\ it consists of the identity
function only. Hence, $\I_1=\V_1$.

Finally, we prove $\C_n \subset \I_n$ for each $n \geq 1$. Consider the formula
$\gamma= \neg X_1\wedge \cdots \wedge \neg X_n$. It is easily seen that $J=[\gamma]_\equiv$ is
a join-irreducible element of $\G_n$. Moreover, every automorphism of $\G_n$ must fix $J$.
To see this, one checks that  $J$ is the only element of $\G_n$ such that
(i) $J$ has height $1$, and (ii) no join-irreducible element of $\G_n$ is
greater than $J$. Since any automorphism of $\G_n$ must preserve properties (i) and (ii)
of $J$, it follows that every such automorphism fixes $J$.
The valuation $\nu\colon \G_n \to \R$
uniquely determined by
\begin{align*}
 \nu(J)&=1 \ ,\\
\nu(x)&=0 \ \text{ for each other join-irreducible } x \in\ \G_n
\end{align*}
is then invariant under the automorphisms of $\G_n$. However, $\nu$ cannot lie in $\C_n$. Indeed, by the
very definition of $\chi_k$, it follows at once that any element of $\C_n$ assigns the same value
to join-irreducible elements of the same height, because each $\chi_k$ has the latter property. This shows
that $\C_n \neq \I_n $. It remains to show that $\C_n\subseteq \I_n$. This holds because
 each automorphism of $\G_n$ carries a join-irreducible of a given height  to a join-irreducible of the same height.
\end{proof}

Finally, we turn to the announced result on the absence of linear relations among the $\chi_k$'s.
\begin{theorem}\label{t:basis}For each integer $n\geq 1$, the set $\{\chi_1,\ldots,\chi_{n+1}\}$ is a basis of $\C_n$. In particular,
$\dim{\C_n}=n+1$.
\end{theorem}
\begin{proof}Let again $\F_n$ be the forest of join-irreducible elements of $\G_n$. As remarked at the beginning
of this section,
 the height of $\F_n$ -- i.e.\ the cardinality of the longest chain in $\F_n$ -- is $n+1$. Let us display such a chain
\[
 c_1 < c_2 <\cdots < c_{n+1} \ .
\]
Suppose that there are real numbers $r_1,\ldots,r_{n+1}\in\R$ such that
\[\tag{$\star$}\label{t:ld}
 r_1\chi_1+\cdots + r_{n+1}\chi_{n+1}=0 \ \ ,
\]
with the intention of showing $r_1=\cdots =r_{n+1}=0$. By Definition  \ref{d:chik}, the evaluation of  (\ref{t:ld}) at $c_i$,
for each $1\leq i\leq n+1$, yields the system of equations
\[
\left\{\begin{array}{ccc}\tag{S}\label{t:S}
r_1 +  r_2  +  \cdots  +  r_i  +  r_{i+1}  +  \cdots  +  r_{n+1} & = & 0\\
& \vdots & \\
r_1 +  2r_2  +  \cdots  +  ir_i  +  ir_{i+1}  +  \cdots  +  ir_{n+1} & = & 0\\
& \vdots & \\
r_1 +  2r_2  +  \cdots  +  ir_i  +  (i+1)r_{i+1}  +  \cdots  +  (n+1)r_{n+1} & = & 0
\end{array}\right.
\]
The determinant of the system (\ref{t:S}) is
\[
\left| \begin{array}{cccccc}
1 & 1 & 1 & \cdots & 1 & 1\\
1 & 2 & 2 & \cdots & 2 & 2\\
1 & 2 & 3 & \cdots & 3 & 3\\
\vdots  &  \vdots  & \vdots  & \ddots  & \vdots  & \vdots\\
1 & 2 & 3 & \cdots & n & n\\
1 & 2 & 3 & \cdots & n & n+1
\end{array} \right| =
\left| \begin{array}{cccccc}
1 & 1 & 1 & \cdots & 1 & 1\\
0 & 1 & 1 & \cdots & 1 & 1\\
0 & 0 & 1 & \cdots & 1 & 1\\
\vdots  &  \vdots  & \vdots  & \ddots  & \vdots  & \vdots\\
0 & 0 & 0 & \cdots & 1 & 1\\
0 & 0 & 0 & \cdots & 0 & 1
\end{array} \right| = 1 \ .
\]
It follows that the system (\ref{t:S}) has a unique solution, namely, $r_1=\cdots=r_{n+1}=0$.
\end{proof}

\begin{remark} The generalised characteristics are integer-valued: their range is contained in the set of integers
$\Z\subseteq\R$. Linear combinations with integer coefficients of generalised characteristics   are again integer-valued.
Therefore, if we write $\C_n^\Z \subseteq \C_n$ for the set
of such linear combinations of generalised characteristics of $\G_n$, then $\C_n^\Z$ has the structure of a $\Z$-module.
In the proof above
of Theorem \ref{t:basis}, the fact that the determinant of the   system \textup{(\ref{t:S})} has value $1$ -- i.e.\ that the
matrix of coefficients of \textup{(\ref{t:S})}
 is \emph{unimodular} -- can be used to prove that $\C_n^\Z$ contains \emph{all} integer-valued valuations of $\C_n$.
\end{remark}

\begin{ack}
We are grateful to the anonymous referee  for a careful
reading of our paper, and for his/her suggestion that the results presented
here may have extensions to other prominent $t$-norm based logics related
to \god{} logic, such as the logic of nilpotent minimum \cite{nm}.
\end{ack}
%
\bibliographystyle{ijuc}
\bibliography{cdm_mvlsc}

\end{document}